\documentclass[12pt]{article}
\usepackage{mathtools,amssymb}
\usepackage{tikz}
\usepackage{amsmath}
\usepackage{amsthm}
\usepackage{geometry}
\usepackage{tgpagella}
\usepackage[title]{appendix}
\usepackage[round]{natbib}
\usepackage{setspace}
\usepackage{float}
\usepackage[ruled,vlined]{algorithm2e}
\usepackage{hyperref}

\newtheorem{defn}{Definition}
\newtheorem{prop}{Proposition}
\newtheorem{exmp}{Example}

\newtheorem{thm}{Theorem}
\newtheorem{claim}{Claim}
\newtheorem{lemma}{Lemma}

\newtheorem{axiom}{Axiom}

\newcommand{\ch}{\operatorname{Ch}}

\title{College Admissions with Scholarship}
\author{Charles Po-Cheng Huang\footnote{Department of Economics, National Central University, No. 300, Zhongda Rd., Zhongli District, Taoyuan City, Taiwan 320317. Contact: cpchuang@ncu.edu.tw. This paper is revised from my dissertation, I want to thank Bumin Yenmez, Rahul Deb, and Utku Ünver for their guidance and support. I thank Fuhito Kojima, Chen-Yu Pan, Kenzo Imamura, and Kentaro Tomoeda for their superb comments, and the participants at the 18th EGSC at Washington University in St. Louis, RCHSS at Academia Sinica, Asia Game Theory Conference 2026.}} 
\date{September 4, 2026}

\begin{document}
	\maketitle
	
	\begin{abstract}
    Merit-based scholarships are widely used to attract high-achieving students, but do they guarantee an improvement in the merit of a university's student pool? We study this question in centralized college admissions using a matching with contracts framework. We first introduce the \emph{scholarship choice rule}, which processes scholarship levels from highest to lowest and selects the highest-merit students at each level. We show that \emph{scholarship feasibility}, \emph{scholarship maximality}, and \emph{no justified envy} uniquely characterize this rule. We further characterize the student-proposing deferred-acceptance mechanism based on this rule as the unique mechanism, up to outcome equivalence, satisfying the extensions of the three axioms, together with individual rationality and strategy-proofness. We then study the welfare effects of scholarship provision. In general matching markets, a university may be matched with lower-merit students after introducing a scholarship. Finally, when all schools use the scholarship choice rule and every student is eligible for every scholarship level a school offers, we show that a common merit ranking is necessary and sufficient for scholarship provision to guarantee, for every student-preference profile satisfying within-school monotonicity, a weak improvement in the scholarship-providing school's student pool.

	\vspace{1cm}
    \end{abstract}

	\section{Introduction}\label{sec:intro}
    Scholarship provision is an important component of higher-education policy in many countries. In Hungary, applicants may seek admission to the same university program under either a state-funded position or a self-financed position. Brazil's University for All Program (ProUni) awards full and partial scholarships at private universities using applicants' national-examination performance together with income eligibility requirements. Chile operates the Academic Excellence Scholarship, which provides tuition support to students graduating in the top ten percent of their secondary-school cohorts, subject to socioeconomic eligibility. Foundation (\emph{vakıf}) universities in Turkey offer 25, 50, or 100 percent tuition waivers based on entrance-examination scores. Taiwanese universities operate dedicated Centers of Admissions and Strategy to manage scholarship programs that target high-scoring students. Although these programs differ in their institutional details, they illustrate the widespread use of academic achievement in allocating financial support for higher education. This prevalence motivates our central question: does introducing scholarships necessarily enable a university to admit a higher-merit student body?
    
    This paper studies merit-based scholarship provision in centralized college admissions markets using the matching with contracts framework of \citet{hatfield2005matching}. We model admission with or without a scholarship as distinct contracts and assume students prefer scholarships to non-scholarship admission at the same institution, a condition we called \emph{within-school monotonicity}. To specify how schools choose among these contracts, we introduce the \emph{scholarship choice rule}. The rule processes scholarship levels from highest to lowest, selecting the highest-merit students at each level up to the capacity of that level, and then fills the school's remaining capacity with the highest-merit students without scholarships. We characterize this rule using three properties: \emph{scholarship feasibility}, \emph{scholarship maximality}, and \emph{no justified envy}. Scholarship feasibility requires the school to respect student, total, and scholarship-level capacities. Scholarship maximality requires that if a contract is rejected despite available capacity at its scholarship level, then the student must have been selected at a higher scholarship level. No justified envy requires that if a student's contract is rejected at a scholarship level and the student is not selected at a higher level, then every student selected at that level must have higher merit. These three properties uniquely characterize the scholarship choice rule (Proposition~\ref{prop:characterization}).
    
    We next move from the choice rule to matching mechanisms. Although the proposed choice rule is not path independent, we construct a path-independent and size-monotonic completion following \citet{hatfield2015hidden}. The completion allows us to apply existing results to obtain stability and strategy-proofness. We then characterize the student-proposing deferred-acceptance mechanism based on the scholarship choice rule, and we show that, on the domain of student preferences satisfying within-school monotonicity, this mechanism is the unique mechanism, up to outcome equivalence, that satisfies the extensions of scholarship feasibility, scholarship maximality, and no justified envy, as well as individual rationality and strategy-proofness (Proposition~\ref{prop:characterizeMech}).
    
    Having established the institutional and mechanism foundations, we turn to the welfare effects of scholarship provision. Our main result shows that, in general matching markets, offering scholarships can backfire. A university may end up matched with lower-merit students after introducing a scholarship program. When a university introduces a scholarship, some students may shift from competing for non-scholarship seats to competing for scholarship seats. This reallocation can create vacancies at other universities and trigger a chain reaction in the matching. A student who would have attended university $s$ without the scholarship may instead choose a competing university, while $s$ awards the scholarship to a student it would have admitted anyway. The resulting vacancies may then be filled by students with lower merit scores than those originally matched with $s$. We provide an explicit example demonstrating this phenomenon (Theorem~\ref{thm:neg}). The negative result applies to any school choice rule induced by school-specific merit rankings, suggesting that additional structure on priorities is required to guarantee scholarship provision to improve the student body. This result relates to \citet{hatfield2016improving}, who show that improving a school in students' preference rankings need not improve the set of students assigned to it according to the school's priority ranking.
    
	We therefore ask what restrictions restore the intended effect of scholarships. We first restrict the domain of the set of contracts by assuming \emph{universal scholarship eligibility}, which requires that every student must be eligible to compete for every scholarship level a school offers. Next, we restrict the priorities ranking by assuming a common ranking: all schools rank students in the same order. Although a common ranking may arise naturally in centralized systems that use national entrance examinations, it remains a restrictive condition. In particular, none of the country examples discussed above satisfies the common ranking condition. Our main welfare result provides a necessary-and-sufficient characterization. When all schools use the scholarship choice rule and universal scholarship eligibility is satisfied, a common ranking is necessary and sufficient to ensure that, for every student-preference profile satisfying within-school monotonicity, introducing a scholarship weakly improves the scholarship-providing school's student pool in terms of merit (Theorem~\ref{thm:commonrank}).
	
	\subsection*{Related Literature}
	The college admission problem was first introduced by \citet{gale1962college}, where they considered a matching problem without contracts and assumed colleges have responsive preferences.\footnote{Responsive preference implies that a college chooses students according to a predetermined ranking up to capacity.} They proposed the well-known deferred-acceptance (DA) algorithm to find a stable matching. \citet{hatfield2005matching} generalized the model by incorporating contracts. They showed that when colleges' choice rules satisfy substitutes or path independence, a stable matching exists (see also \citealp{aygun2013matching}; \citealp{hatfield2008matching}; \citealp{chambers2017choice}). Subsequent literature weakened the assumption of substitutes by using weak substitutes or substitute completion (\citealp{hatfield2010substitutes}; \citealp{hatfield2015hidden}). Moreover, \citet{yenmez2018college} weakened path independence by introducing a path-independence modification. \citet{Hatfield2021Stability} introduce observable substitutability and related conditions under which the cumulative-offer mechanism is the unique stable and strategy-proof mechanism. Their results allow us to use the completion of the scholarship choice rule to establish stability and strategy-proofness and to characterize the associated mechanism on our restricted preference domain.
	
	\citet{abizada2016stability}, \citet{abizada2026college}, and \citet{afacan2020graduate} study college admissions in which financial support is part of the assignment, the first two through budget constraints that generate complementarities between students and the last by treating support options as contracts in a graduate admission problem. Their focus is primarily on the existence of stable and strategy-proof mechanisms under these constraints. \citet{Biro2025Large} study college admissions in which students may attend the same program under different financial terms. Using the Hungarian admission system, they analyze the size and composition of the stable core and compare merit-based and non-merit-based stable allocations. In contrast to this literature, we study how introducing scholarships affects the merit composition of a school's student pool. Our scholarships setting also generates complementarities because students compete for limited scholarship seats at each level, making the selection of one student affecting the selection of others.

	Our scholarship choice rule is related to the meritorious horizontal choice rule of \citet{sonmez2022affirmative}. While their reserve protections are minimum guarantees associated with agents' traits, scholarship levels in our model are contractual terms for which students have preferences. The sequential structure of our rule is also related to the slot-specific framework of \citet{kominers2016matching}, in which slots with potentially different priorities are filled according to a precedence order. Unlike their fixed slot structure, non-scholarship admission in our model has no separate capacity and instead uses the total enrollment capacity remaining after scholarship contracts are chosen. This residual structure belongs to the slot-specific capacity-transfer framework of \citet{avataneo2021slot}. Rather than pursuing such generality, we focus on the particular choice rule induced by scholarship allocation and provide an axiomatic characterization that uniquely determines it.

    Our mechanism characterization builds on \citet{DOGAN2025106057}. They provide a general method for extending solitary axioms that characterize institutional choice rules into axioms that characterize the corresponding deferred-acceptance mechanism. Our three scholarship axioms are solitary in their sense. However, we cannot directly apply their method to characterize the mechanism since our scholarship choice rule is not path independent, and their characterization required the choice rule to be path independent. 

	Our paper connects to the literature on respecting improvements, which was introduced by \citet{balinski1999tale}. \citet{afacan2020graduate} showed that his mechanism for graduate admission with financial support respects student improvements, while \citet{avataneo2021slot} showed that the cumulative offer mechanism under slot-specific priorities with capacity transfers also respects students' improvements. \citet{hatfield2016improving} instead studied school improvements, showing that stable mechanisms respect school quality improvements only approximately in large markets. Our framework is closer to the latter question, but models school improvement through scholarship provision, rather than through an improvement in students' preferences over schools.
	
	The rest of the paper is organized as follows. Section~\ref{sec:model} introduces the model and definitions. Section~\ref{sec:choicerule} introduces and characterizes the scholarship choice rule, and provides the mechanism characterization. Section~\ref{sec:welfare} studies the welfare effects of scholarship provision, establishing the general negative result and the common-ranking characterization. Section~\ref{sec:conclude} concludes the paper. The proofs of the results are left to Appendix~\ref{sec:appen}.
	\section{Model}\label{sec:model}
    
	Let $\mathcal{I}$ be a nonempty finite set of students, and $\mathcal{S}$ be a nonempty finite set of schools. Let $M = \{0,m_1,m_2,\ldots ,m_K\}$ be the finite set of possible scholarship levels, where $0 < m_1 < m_2 < \ldots <m_K$ and $m = 0$ represents admission without a scholarship. A contract $x$ is a triple $(i,s,m) \in \mathcal{I} \times \mathcal{S} \times M$, and $\mathcal{X} \subseteq \mathcal{I} \times \mathcal{S} \times M$ denotes the set of available contracts. We assume that $(i,s,0) \in \mathcal{X}$ for every $i \in \mathcal{I}$ and $s \in \mathcal{S}$, that is, every student-school pair has a basic contract. For each $x \in \mathcal{X}$, let $\iota(x)$ be the student associated with contract $x$; similarly, let $\sigma(x)$ be the school associated with contract $x$, and let $\mathbf{m}(x)$ be the scholarship term associated with contract $x$. Furthermore, we denote $\mathcal{X}_i := \{x\in \mathcal{X}: i = \iota(x)\}$; that is, $\mathcal{X}_i$ is the set of contracts that student $i$ is involved in. Similarly, $\mathcal{X}_s := \{x\in \mathcal{X}: s = \sigma(x)\}$ is the set of contracts associated with school $s$, and $\mathcal{X}_m := \{x\in \mathcal{X}: m = \mathbf{m}(x)\}$ is the set of contracts associated with scholarship level $m$.
	
	Each student $i \in \mathcal{I}$ has a preference $\succ_i$ over $\mathcal{X}_i \cup \{\emptyset\}$, which is the set of contracts associated with $i$ and an outside option. We impose \textbf{within-school monotonicity}: for every student $i$, every school $s$, and every pair of available contracts $(i,s,m_\ell),(i,s,m_k) \in \mathcal X_i$ with $m_\ell > m_k$, $(i,s,m_\ell) \succ_i (i,s,m_k)$. Let $\succeq_i$ be the weak preference relation induced by $\succ_i$. Let $\mathcal R_i^{WM}$ be the set of preference relations over $\mathcal X_i\cup\{\emptyset\}$ whose strict parts satisfy within-school monotonicity, and define $\mathcal R^{WM}:=\prod_{i\in\mathcal I}\mathcal R_i^{WM}.$ Throughout the paper, the student-preference profile $\succeq_{\mathcal I}:=(\succeq_i)_{i\in\mathcal I}$ belongs to $\mathcal R^{WM}$. For any set of contracts $X\subseteq\mathcal X$, let $X_i:=\{x\in X:i=\iota(x)\}$ and define the choice rule induced by $\succ_i$ as $\ch_i(X)=\max_{\succ_i}\bigl(X_i\cup\{\emptyset\}\bigr).$ Thus, student $i$ chooses her most-preferred acceptable contract in $X_i$ whenever one exists and chooses $\emptyset$ otherwise. Because the student has unit capacity and her choice is induced by a fixed strict ranking over individual contracts and the outside option, $\ch_i$ is a responsive choice rule. On the other hand, each school $s$ has a choice rule $\ch_s$ such that, for every set of contracts $X\subseteq\mathcal X_s$, $\ch_s(X)\subseteq X$.
	
	For a school $s \in \mathcal{S}$, each student $i$ is endowed with a distinct merit score $\pi_s(i) \in \mathbb{R}_+$. Thus, $\pi_s$ induces a strict ranking of contracts at each scholarship level through the associated students. Since the merit score is school-specific, schools can have different priority rankings over students. We assume every student is acceptable to every school. Let $\iota(X_s) = \{\iota(x):x \in X_s\}$; that is, the set of students which school $s$ has contracts with under a set of contracts $X$. Each school $s$ has a capacity $q_s\in\mathbb N$ that regulates the total number of contracts school $s$ can choose. Let $M^s \subseteq \{m_1,\ldots,m_K\}$ be the set of scholarship levels offered by school $s$. We assume that for all $x \in \mathcal X_s$, $\mathbf{m}(x) \in M^s \cup \{0\}$; that is, contracts involving school $s$ exist only at scholarship levels offered by $s$ or the no-scholarship level. Furthermore, school $s$ has an integer capacity $q^{m_k}_s\in\mathbb N_0$ for each level $m_k \in M^s$. We assume that $\sum_{m_k \in M^s} q^{m_k}_s \leq q_s$; that is, the total scholarship capacity at $s$ cannot exceed its total capacity. We do not impose a separate capacity at the zero-scholarship level; zero-scholarship contracts use the total capacity remaining after positive-scholarship contracts are chosen.
	
	Let $\ch_\mathcal{I} := (\ch_i)_{i \in \mathcal{I}}$ and $\ch_\mathcal{S} := (\ch_s)_{s \in \mathcal{S}}$ denote the choice rule profiles of the students and of the schools, respectively. A market is a tuple $\langle \mathcal{I}, \mathcal{S}, \mathcal{X}, \ch_\mathcal{I}, \ch_\mathcal{S} \rangle$. A matching is a set of contracts $\mu\subseteq\mathcal{X}$. For each school $s$, let $\mu_s:=\mu\cap\mathcal X_s$. We consider many-to-one matching, where a student can get at most one contract under a matching, while a school may have a set of contracts. We say a matching $\mu$ is \emph{feasible} if $|\mu\cap\mathcal X_i| \le 1$, $|\mu_s| \leq q_s$, and $|\{x \in \mu_s : \mathbf{m}(x) = m_k\}| \leq q^{m_k}_s$, for each student $i$, each school $s$, and each scholarship level $m_k \neq 0$. For any feasible matching, let $\mu_i$ denote the unique contract in $\mu\cap\mathcal X_i$ when this set is nonempty, and let $\mu_i=\emptyset$ otherwise. We now define stability:
	
	\begin{defn}\label{def:stable}
		A matching $\mu$ is \textbf{stable} if the following two conditions hold:
		\begin{itemize}
			\item[1.](Individual rationality) For every student $i$, $\ch_i(\mu) = \mu_i$, and for all school $s$, $\ch_s(\mu)=\mu_s$,
			\item[2.](No blocking contracts) There does not exist a school $s$ and a nonempty set of contracts $Z \subseteq \mathcal{X}_s \setminus \mu$ such that $Z \subseteq \ch_s(\mu \cup Z)$, and for every student $i \in \iota(Z)$, $Z_i = \{\ch_i(\mu \cup Z)\}$.
		\end{itemize}
	\end{defn}

	Let $\mathcal M$ denote the set of feasible matchings. A matching mechanism on $\mathcal R^{WM}$ is a function $\varphi:\mathcal R^{WM}\to\mathcal M.$ For every preference profile $\succeq_{\mathcal I}$, let $\varphi_i(\succeq_{\mathcal I})$ denote the contract assigned to student $i$, with $\varphi_i(\succeq_{\mathcal I})=\emptyset$ if $i$ is unmatched. 
	
	We introduce the main mechanism of interest. The student-proposing deferred-acceptance (DA) mechanism based on $\ch_\mathcal{S}$ is defined through the following algorithm \citep{gale1962college}.

	\noindent\textbf{Student-Proposing Deferred-Acceptance Algorithm Based on $\ch_\mathcal{S}$}
	
	\begin{enumerate}
		\item[Step 1:] Each student proposes her most preferred acceptable contract to the associated school (if there is none, she is assigned the null contract $\emptyset$). Each school $s$ considers the proposals it receives at this step, say $X^1_s$, tentatively accepts $\ch_s(X^1_s)$, and permanently rejects $X^1_s \setminus \ch_s(X^1_s)$. If there is no rejection by any school at this step, then stop and return $\bigcup_{s \in \mathcal{S}} \ch_s(X^1_s)$. Otherwise, go to Step 2.
		\item[Step $t \geq 2$:] Each student whose contract was rejected in the previous step proposes her next most preferred acceptable contract to the associated school (if there is none, she is assigned the null contract). Each school $s$ considers its tentatively accepted contracts from the previous step together with the proposals it receives at this step, say $X^t_s$, tentatively accepts $\ch_s(X^t_s)$, and permanently rejects $X^t_s \setminus \ch_s(X^t_s)$. If there is no rejection by any school at this step, then stop and return $\bigcup_{s \in \mathcal{S}} \ch_s(X^t_s)$. Otherwise, go to Step $t+1$.
	\end{enumerate}

	The algorithm stops at a finite step because no student proposes any contract more than once and the set of contracts $\mathcal{X}$ is finite. The DA mechanism based on $\ch_\mathcal{S}$ chooses, at each preference profile, the matching defined by the acceptances at the last step of the DA algorithm based on $\ch_\mathcal{S}$.

	A matching mechanism $\varphi$ satisfies a matching axiom if and only if, for each preference profile $\succeq_{\mathcal I} \in \mathcal R^{WM}$, $\varphi(\succeq_{\mathcal I})$ satisfies the axiom. The following is a well-known axiom on matching mechanisms rather than on matchings, and it requires that no student can benefit from misreporting her preferences.

	\begin{defn}[Strategy-proofness]\label{def:strategy-proof}
		A matching mechanism $\varphi$ is \emph{strategy-proof on $\mathcal R^{WM}$} if, for every preference profile $\succeq_{\mathcal I}=(\succeq_i,\succeq_{-i})\in\mathcal R^{WM}$, every student $i\in\mathcal I$, and every admissible alternative report $\widehat{\succeq}_i\in\mathcal R_i^{WM}$,
		\[
		\varphi_i(\succeq_i,\succeq_{-i})
		\succeq_i
		\varphi_i(\widehat{\succeq}_i,\succeq_{-i}).
		\]
	\end{defn}

	In words, truthful reporting weakly dominates every unilateral misreport that satisfies within-school monotonicity, holding the reports of all other students fixed.
	
	When a school is comparing two sets of contracts, or two matchings, we assume the school cares only about the quality of the students, not the cost of the scholarships it awards. In other words, scholarships are tools to attract students; once the scholarship contracts offered by a school are given, the school considers the quality of the set of students. To compare the quality of two sets of students, we use the following definition:
	\begin{defn}[Gale-merit domination]\label{def:Galedominate}
		Given $A, B \subseteq \iota(X_s)$, let $A = \{a_1,a_2,\ldots,a_{|A|}\}$ and $B = \{b_1,b_2,\ldots,b_{|B|}\}$ be enumerated such that:\footnote{Notice that $A$ and $B$ are two sets of students that have contracts associated with school $s$. $A$ and $B$ are not two sets of contracts.}
		\begin{itemize}
			\item[(1.)] for each $j,k \in \{1,2,\ldots,|A|\}$, $j\leq k \Rightarrow \pi_s(a_j) \geq \pi_s(a_k)$,
			\item[(2.)] for each $j,k \in \{1,2,\ldots,|B|\}$, $j\leq k \Rightarrow \pi_s(b_j) \geq \pi_s(b_k)$. 
		\end{itemize}
		$A$ \textbf{Gale-dominates} $B$, or $A \trianglerighteq^G B$, if $|A| \geq |B|$ and, for each $j \in \{1,\ldots,|B|\}$, $\pi_s(a_j) \geq \pi_s(b_j)$.
	\end{defn} 
	Definition \ref{def:Galedominate} ranks two sets of students $A$ and $B$ that school $s$ has contracts with by their merit scores from high merit to low merit. It then compares the size of the sets and the merit score at each rank until the last student in set $B$ is reached. Next, we say $A$ \textbf{strictly Gale-dominates} $B$, or $A \triangleright^G B$,  if $A \trianglerighteq^G B$ and $B \ntrianglerighteq^G A$.
	
	\section{The Scholarship Choice Rule}\label{sec:choicerule}
    We first introduce the key property for any choice rule in matching with contracts, namely \emph{path independence}. \citet{aizerman1981general} showed that path independence is equivalent to the conjunction of two properties: \emph{substitutability}, which says that adding another contract should not cause a previously rejected contract to become chosen, and \emph{consistency}, which says that deleting contracts that are not chosen will not affect the chosen set.\footnote{Formally, a choice rule $\ch$ is \emph{substitutable} if, for all $x, x' \in \mathcal{X}$ and $Y \subseteq \mathcal{X}$, $x \notin \ch(Y \cup \{x\}) \Rightarrow x \notin \ch(Y \cup \{x,x'\})$; and $\ch$ is \emph{consistent} if, for all $X, Y \subseteq \mathcal{X}$, $\ch(X) \subseteq Y \subseteq X \Rightarrow \ch(Y) = \ch(X)$. We invoke these two properties individually in the proofs in Appendix~\ref{sec:appen}.} We define path independence as follows:

	\begin{defn}\label{def:pathind}
			A choice rule $\ch$ is \textbf{path independent} if, for each $X,Y \subseteq \mathcal{X}$, $$\ch(X\cup Y) = \ch(X\cup \ch(Y)).$$ 
	\end{defn}

	In words, a choice rule satisfies path independence if the choice of the set of contracts remains unchanged when we first choose from an arbitrary subset of contracts, combine the result with other contracts, and choose again. \citet{chambers2017choice} showed that the student-proposing deferred-acceptance mechanism produces a stable matching when schools and students all have path-independent choice rules. Since the students' choice rule $\ch_i$ is a \emph{responsive choice rule}, it is known to be path independent. Our goal is to find a choice rule for schools that is natural in this scholarship context and path independent.
    
	Before turning to the details of the choice rule, let us define another desired property, \emph{size monotonicity}. 
	
	\begin{defn}\label{def:sizemon}
		A school choice rule $\ch_s$ is \textbf{size monotonic} if, for all
		$X\subseteq Y\subseteq\mathcal X_s$,
		\[
		|\ch_s(X)|\leq |\ch_s(Y)|.
		\]
	\end{defn}

	In words, if the size of the set of contracts that school $s$ can choose weakly increases, then the size of the chosen contracts cannot decrease. This is also known as the law of aggregate demand \citep{hatfield2005matching}. 
	
    To come up with a simple choice rule for schools, we consider three properties we regard as desirable for a school choice rule. The first axiom requires the school's choice to respect the relevant feasibility constraints.

    For any set of contracts \(Y\subseteq \mathcal X_s\), let \(Y^m:=\{y\in Y:\mathbf m(y)=m\}\) and \(Y_i:=\{y\in Y:\iota(y)=i\}.\) Define the family of feasible subsets of a set of contracts \(X\subseteq\mathcal X_s\) by
    \[
    \mathcal F_s(X) :=
    \left\{
    Y\subseteq X:
    \begin{array}{l}
    |Y_i|\leq 1 \text{ for every student }i,\\
    |Y|\leq q_s,\\
    |Y^{m_k}|\leq q_s^{m_k}
       \text{ for every }m_k\in M^s
    \end{array}
    \right\}.
    \]

    Because the capacity of the non-scholarship contract is determined
    residually, for every feasible set \(Y\) define the effective capacity
    of scholarship level \(m\) by
    \[
    \bar q_s^m(Y)
    :=
    \begin{cases}
    q_s^m, & \text{if }m\in M^s,\\[1mm]
    q_s-\displaystyle\sum_{\ell\in M^s}|Y^\ell|,
        & \text{if }m=0.
    \end{cases}
    \]
    Thus, \(\bar q_s^0(Y)\) is the number of seats remaining after the
    scholarship contracts in \(Y\) have been counted.
    
    \begin{axiom}[Scholarship Feasibility]\label{axiom:sch-feasibility}
    For every \(X\subseteq\mathcal X_s\), \(\ch_s(X) \in \mathcal F_s(X).\)
    \end{axiom}
    
    \emph{Scholarship feasibility} requires that the school choose at most one contract for each student, respect its total capacity, and respect the
    capacity associated with every positive scholarship level.
    
    \begin{axiom}[Scholarship Maximality]\label{axiom:sch-maximal}
    For every \(X\subseteq\mathcal X_s\), suppose $x\in X\setminus \ch_s(X)$ and suppose that $\left| \{y\in \ch_s(X):\mathbf m(y)=\mathbf m(x)\} \right| < \bar q_s^{\mathbf m(x)}(\ch_s(X))$. Then, there exists \(z\in \ch_s(X)\) such that
    \[
    \iota(z)=\iota(x) \quad \text{and} \quad \mathbf m(z)>\mathbf m(x).
    \]
    \end{axiom}
    
    \emph{Scholarship maximality} says that if a contract is rejected even though there is unused capacity at its scholarship level, then the associated student must have been selected under a contract with a higher scholarship level. In particular, when \(\mathbf m(x)=0\), unused capacity means that the school has not filled its total capacity.
    
    \begin{axiom}[No Justified Envy]\label{axiom:nje}
    For every \(X\subseteq\mathcal X_s\), suppose that \(x\in X\setminus \ch_s(X)\), no \(z\in \ch_s(X)\) satisfies $\iota(z)=\iota(x)$ with $\mathbf m(z)>\mathbf m(x),$ and $\left|\{y\in \ch_s(X):\mathbf m(y)=\mathbf m(x)\}\right| = \bar q_s^{\mathbf m(x)}(\ch_s(X)).$
	Then, for every \(y\in\ch_s(X)\) with \(\mathbf m(y)=\mathbf m(x)\),
    \[
    \pi_s(\iota(y))>\pi_s(\iota(x)).
    \]
    \end{axiom}
    
    \emph{No justified envy} says that if a student is not selected at a higher scholarship level and her contract is rejected when its scholarship level is full, then every student selected at that level has a higher merit score than the rejected student.

    Each of the three axioms is solitary in the sense of \citet{DOGAN2025106057}. For a fixed set of contracts \(X\), whether a proposed choice \(Y\subseteq X\) satisfies any of the three axioms can be determined using only \(X\) and \(Y\); choices from other choice problems are irrelevant. Inspired by these axioms, we introduce the scholarship choice rule, denoted by \(\ch_s^{Sc}\).
    
    \begin{algorithm}[H]
    \caption{The Scholarship Choice Rule $\ch^{Sc}_s$}
	    \KwIn{Set of contracts $X_s \subseteq \mathcal{X}_s$}
    \KwOut{Chosen contracts $\ch^{Sc}_s(X_s)$}

    \textbf{Initialization:}\\
    Order scholarship levels: for all $m_k \in M^s \subseteq \{m_1, m_2, \ldots, m_K\}$, order the scholarship levels from highest to lowest, denoted as $M_{\text{ordered}}$ \;
    Order students by merit: order $\iota(X_s)$ according to $\pi_s$ from highest to lowest, denoted as $I_{\text{ordered}}$\;
    $C_f \leftarrow \emptyset$, $X_{\text{avail}} \leftarrow X_s$\;

    \BlankLine
    \textbf{Fill scholarship seats:}\\
    \ForEach{$m \in M_{\text{ordered}}$}{
	   $X_m \leftarrow \{x \in X_{\text{avail}}: \mathbf{m}(x) = m\}$\; 
	   $I^*_m \leftarrow$ top $\min(|\iota(X_m)|, q^m_s)$ students from $\iota(X_m)$\;
	   $C_m \leftarrow \{x \in X_m: \iota(x) \in I^*_m\}$\;
	   $C_f \leftarrow C_f \cup C_m$\;
	   $X_{\text{avail}} \leftarrow X_{\text{avail}} \backslash \{x: \iota(x) \in I^*_m\}$\;
    }

    \BlankLine
    \textbf{Fill remaining seats:}\\
    $q_{\text{remain}} \leftarrow q_s - |C_f|$\;
    $X_0 \leftarrow \{x \in X_{\text{avail}}: \mathbf{m}(x) = 0\}$\;
    $I^*_0 \leftarrow$ top $\min(|\iota(X_0)|, q_{\text{remain}})$ students from $\iota(X_0)$\;
    $C_0 \leftarrow \{x \in X_0: \iota(x) \in I^*_0\}$\;
    $C_f \leftarrow C_f \cup C_0$\;

    \Return{$C_f$}
    \end{algorithm}

	In other words, the scholarship choice rule $\ch^{Sc}_s$ searches the available contracts in a greedy way. First, it starts with the highest scholarship level and chooses the contract associated with the highest-merit remaining student up to the level capacity. Then, it proceeds to the next highest scholarship level in the same way. Lastly, it goes to the non-scholarship contracts. 
	
	The first result shows that $\ch_s^{Sc}$ is the unique choice rule that satisfies Axiom~\ref{axiom:sch-feasibility}, Axiom~\ref{axiom:sch-maximal}, and Axiom~\ref{axiom:nje}.
    
    \begin{prop}\label{prop:characterization}
    A choice rule is scholarship feasible, scholarship maximal, and free of justified envy if and only if it is the scholarship choice rule \(\ch_s^{Sc}\).
    \end{prop}
       
	Although $\ch^{Sc}_s$ is intuitive, it fails path independence. We therefore construct a completion satisfying path independence. As \citet{hatfield2015hidden} point out, if a profile of choice rules $(\ch_s)_{s\in\mathcal S}$ has such a completion, then every stable outcome with respect to the completion is also stable with respect to the original profile (see their Lemma 1). \citet{yenmez2018college} generalize this approach to many-to-many matching and show that deferred acceptance produces a stable matching when student choice rules are path-independent and school choice rules have path-independent modifications (see their Theorem 1). We first define a completion.

    \begin{defn}[Choice rule completion]\label{def:completion}
			A school choice rule $\overline{\ch}_s$ is a \textbf{completion} of a choice rule $\ch_s$ if, for every set of contracts $X \subseteq \mathcal{X}_s$, either:
		\begin{itemize}
			\item[1.] There exist distinct contracts $x, y \in \overline{\ch}_s (X)$ such that $\iota(x) = \iota(y)$, or
			\item[2.] $\overline{\ch}_s(X) = \ch_s(X)$. 
		\end{itemize}
	\end{defn}
    
	Proposition~\ref{prop:PathI} gives both the failure of path independence and the desired completion.

	\begin{prop}\label{prop:PathI}
		$\ch^{Sc}_s$ is not path independent, but $\ch^{Sc}_s$ has a path-independent and size-monotonic completion, denoted by $\overline{\ch}^{Sc}_s$.
	\end{prop}

	The failure of path independence comes from the interaction between scholarship levels: choosing a student at a higher level removes that student's contracts from lower-level competitions. The completion removes this interaction by conducting each level's merit competition independently.

	\subsection*{Characterizing Matching Mechanisms}
	We follow \citet{DOGAN2025106057} to characterize the DA mechanism based on the scholarship choice rule. For each $k\in\{F,M,N\}$, define the correspondence $\psi_s^k:2^{\mathcal X_s}\rightrightarrows 2^{\mathcal X_s},$ where $F$ refers to scholarship feasibility (Axiom~\ref{axiom:sch-feasibility}), $M$ refers to scholarship maximality (Axiom~\ref{axiom:sch-maximal}), and $N$ refers to no justified envy (Axiom~\ref{axiom:nje}). For every $X\subseteq\mathcal X_s$, let
	\[
	\psi_s^k(X)
	:=
	\left\{
	Y\subseteq X:Y\text{ satisfies axiom }k\text{ at }X
	\right\}.
	\]
	Thus, $\psi_s^k(X)\subseteq 2^X$ is the collection of permissible subsets of $X$ under axiom $k$. Proposition~\ref{prop:characterization} can then be rewritten as
	\[
	\psi^F_s(X)\cap\psi^M_s(X)\cap\psi^N_s(X)
	=
	\{\ch^{Sc}_s(X)\}
	\]
	for every $X\subseteq\mathcal X_s$ and every $s\in\mathcal S$.

	Next, we define the demand for school $s\in\mathcal S$. For any student-preference profile $\succeq_{\mathcal I}$ and any feasible matching $\mu$, let
	\[
	D_s(\mu,\succeq_{\mathcal I})
	:=
	\left\{
	x\in\mathcal X_s:x\succeq_{\iota(x)}\mu_{\iota(x)}
	\right\}.
	\]
	For each $k\in\{F,M,N\}$, a matching $\mu$ satisfies the \textbf{extension of axiom $k$} at school $s$ under $\succeq_{\mathcal I}$ if
	\[
	\mu_s\in\psi_s^k\bigl(D_s(\mu,\succeq_{\mathcal I})\bigr).
	\]
	A matching mechanism $\varphi$ satisfies the extension of axiom $k$ if, for every student-preference profile $\succeq_{\mathcal I}$ and every school $s$, the matching $\mu=\varphi(\succeq_{\mathcal I})$ satisfies this condition.

	We are now ready to provide our characterization result:
	\begin{prop}\label{prop:characterizeMech}
		Given the student-preference domain $\mathcal{R}^{WM}$, a matching mechanism $\varphi$ satisfies
		\begin{enumerate}
			\item the extension of scholarship feasibility,
			\item the extension of scholarship maximality,
			\item the extension of no justified envy,
			\item individual rationality, and
			\item strategy-proofness,
		\end{enumerate}
		if and only if it is outcome equivalent to the student-proposing DA mechanism based on $(\ch^{Sc}_s)_{s \in \mathcal{S}}$.
	\end{prop}

	Proposition~\ref{prop:characterizeMech} shows that the three axioms introduced above have an exact counterpart at the market level. At a matching $\mu$, the contracts a school could plausibly have obtained are those that students weakly prefer to their own assignments, so the natural market-level analogue of an axiom is to require the school's assignment to be what the axiom prescribes from that set. Because Proposition~\ref{prop:characterization} shows that the three axioms together admit a unique solution, this requirement leaves the school no discretion: its assignment must be exactly what the scholarship choice rule would have selected from everything within its reach. 

	\section{Welfare Analysis}\label{sec:welfare}
	
	In this section, we investigate whether introducing a scholarship enables a school to recruit a better set of students. Fix a school $s_k$ and a scholarship level $m\notin M^{s_k}$. Initially, $s_k$ does not offer level $m$. When $s_k$ introduces this scholarship level, its set of scholarship levels becomes $M^{s_k}\cup\{m\}$; it assigns the level a capacity $q^m_{s_k}>0$ satisfying $q^m_{s_k}+\sum_{\ell\in M^{s_k}}q^\ell_{s_k}\leq q_{s_k}$; and the contracts $\{(i,s_k,m):i\in\mathcal I\}$ are added to the market. The total capacity $q_{s_k}$ remains fixed, so filling seats at the new scholarship level reduces the capacity remaining for admission without a scholarship. All other contracts, capacities, and school priorities remain unchanged.
	
	Let $\succeq^A_{\mathcal I}$ denote a student-preference profile after the scholarship level is introduced. The corresponding profile $\succeq_{\mathcal I}$ before the introduction is obtained by deleting the contract $(i,s_k,m)$ from student $i$'s preference while leaving the ranking of all other contracts unchanged, for every student $i \in \mathcal{I}$. Let $\mu$ and $\mu'$ denote the student-proposing DA outcomes under $\succeq_{\mathcal I}$ and $\succeq^A_{\mathcal I}$, respectively. All other schools $s\neq s_k$ may offer any scholarship levels allowed by the model.
	
	Our first theorem shows that introducing a scholarship level does not generally guarantee that $s_k$ will be matched with a weakly better set of students in terms of merit. We construct a market in which the scholarship harms student quality at $s_k$ under Gale domination. This result holds for any school choice rule induced by the schools' priorities.
	
	\begin{thm}\label{thm:neg}
		Let $\mu$ be the stable matching generated from the student-proposing DA mechanism when $s_k$ does not provide a scholarship, and let $\mu'$ be the stable matching generated from the same mechanism when $s_k$ provides a scholarship.
		
		There exists a market such that, for any choice rule induced by the priority scores, $\iota(\mu_{s_k}) \triangleright^G \iota(\mu'_{s_k})$.
	\end{thm}    
	
	In other words, we construct a market where as long as the choice rule follows the priority ranking, $\iota(\mu_{s_k})$ strictly Gale-dominates $\iota(\mu'_{s_k})$. Since we cannot obtain a positive result guaranteeing that $s_k$ is matched with a better set of students after providing the scholarship, we turn to restricting the priority ranking of the schools. Below, we seek the maximal domain of the priority ranking of the schools such that providing the scholarship guarantees a better set of students in terms of merit.
	
	\subsection{Common Priority Ranking}
    To begin with, we consider the common priority ranking. It turns out that a common priority ranking is necessary and sufficient for scholarship provision to guarantee a better set of students. We define common priority ranking as follows:

    \begin{defn}[Common priority ranking]
    The school-priority profile has a \emph{common (priority) ranking} if, for every pair of students \(i,j\) and every pair of schools \(s,t\),
    \[
    \pi_s(i) > \pi_s(j)
    \quad\Longleftrightarrow\quad
    \pi_t(i) > \pi_t(j).
    \]
    \end{defn}

	For the characterization below, we also impose \emph{universal scholarship eligibility}: for every school $s\in\mathcal S$, every scholarship level $m\in M^s$, and every student $i\in\mathcal I$, $(i,s,m)\in\mathcal X.$ Thus, every student is eligible for every scholarship level offered by a school, while the scholarship choice rule determines who receives the seats. This assumption is not needed for the results in Section~\ref{sec:choicerule}. Theorem~\ref{thm:neg} remains valid even under this assumption, since its counterexample makes the new scholarship level available to both students.
    
    \begin{thm}\label{thm:commonrank}
		Fix $\mathcal I$, $\mathcal S$, and $(\pi_s)_{s\in\mathcal S}$. The school-priority profile has a common ranking if and only if, for every market constructed as above and every $\succeq^A_{\mathcal I}\in\mathcal R^{WM}$,
		\[
		\iota(\mu'_{s_k})\trianglerighteq^G\iota(\mu_{s_k}).
		\]
    \end{thm}

	Theorem~\ref{thm:commonrank} identifies the maximal domain of school priorities on which a scholarship guarantees $s_k$ a weakly better set of students: every school must rank students in exactly the same way. The force of common ranking is that it collapses the mechanism into a serial dictatorship, in which students are called in decreasing order of merit and each takes her favorite remaining contract. This lets us introduce the new scholarship contracts one student at a time, so that each step changes the options of a single student. If she declines the scholarship, nothing moves. If she takes it, she vacates her former seat and sets off a single chain of moves that runs strictly downward through the merit order. Along that chain a student can leave $s_k$ only after a higher-merit student has entered it, so departures pair off with entrants and the entrant is always the stronger of the two. Each step therefore weakly improves $s_k$, and Gale domination is transitive. 
	
	Once schools rank students differently, there is no single order for the chain to follow: the scholarship may draw in the student $s_k$ ranks lowest while the student it displaces is absorbed by a rival that ranks her highly, and the pairing fails. Universal scholarship eligibility protects the same order from a different direction: because the scholarship choice rule awards scholarship seats before it fills the school's remaining non-scholarship capacity, a level open to only some students would let a lower-merit student claim a scholarship seat and crowd out a higher-merit applicant who is ineligible for that level and can compete only for a non-scholarship seat. We provide an example which demonstrates that scholarship provision can lead to worse student pool when universal scholarship eligibility is not assumed:

	\begin{exmp}
		Consider two students, $i$ and $j$, and two schools, $s_a$ and $s_b$. Let $q_{s_a} = q_{s_b} = 1$, and both schools use the scholarship choice rule. Assume both schools have the same priority ranking: $ \pi(i) > \pi(j).$ Consider $s_a$ already offers scholarship level $h$, with $q_{s_a}= q^h_{s_a} = 1$, and only the lower-merit student $j$ is eligible for this scholarship. $s_b$ offers no scholarship initially, and it introduces scholarship level $m$, with $q^m_{s_b} = 1$, for both students. We will focus on the student school $s_b$ matches to with/without scholarship level $m$.

		Let the preferences be
		$$
		\begin{aligned}
			i:\quad & (i,s_a,0) \succ_i(i,s_b,m) \succ_i(i,s_b,0) \succ_i \emptyset,\\
			j:\quad &(j,s_b,m) \succ_j(j,s_a,h) \succ_j(j,s_a,0) \succ_j(j,s_b,0) \succ_j \emptyset.
		\end{aligned}
		$$
		
		When $s_b$ does not introduce scholarship level $m$, we can delete the two level-$m$ contracts from the preference profiles. Both students first apply to $s_a$: student $i$ applies without a scholarship, while student $j$ applies for level $h$. The scholarship choice rule fills level $h$ first, so $s_a$ holds $(j,s_a,h)$ and rejects $(i,s_a,0)$. Student $i$ then applies to $s_b$ and $s_b$ accepts. Therefore,
		$$
		\mu_{s_a}=\{(j,s_a,h)\},
		\qquad
		\mu_{s_b}=\{(i,s_b,0)\}.
		$$

		When $s_b$ introduces level $m$, student $i$ still applies first to $s_a$ under $(i,s_a,0)$, while student $j$ now applies to her most-preferred contract $(j,s_b,m)$. Each school receives one application and accepts it. Hence,
		$$
		\mu'_{s_a}=\{(i,s_a,0)\},
		\qquad
		\mu'_{s_b}=\{(j,s_b,m)\}.
		$$

		$s_b$ matches to a worse student when introducing the scholarship contract.
	\end{exmp}
	
	Although common ranking might be plausible in the settings with a centralized exam system, it remains restrictive, since schools cannot weight subjects differently. If the exam score decomposes into a Math score and an English score, common ranking requires subject weights that every school agrees on.
	
	\section{Conclusion}\label{sec:conclude}

	Giving scholarships is one of the most popular methods colleges use to attract higher-achieving students. In this paper, we study scholarships based on academic merit. To capture how schools choose students across scholarship levels, we propose three axioms: \emph{scholarship feasibility}, \emph{scholarship maximality}, and \emph{no justified envy}, and we provide the unique choice rule, the \emph{scholarship choice rule}, that satisfies these axioms. Although this choice rule is not path independent, it admits a path-independent and size-monotonic completion, which allows us to establish stability and strategy-proofness through existing matching with contracts results. We then extend the choice rule characterization to the mechanism level. On the domain of student preferences satisfying within-school monotonicity, the student-proposing deferred-acceptance mechanism based on the scholarship choice rule is the unique mechanism, up to outcome equivalence, satisfying individual rationality, strategy-proofness, and the extensions of scholarship feasibility, scholarship maximality, and no justified envy.
	
	Given the institutional foundation for the choice rule, we answer the question of whether merit-based scholarships enable colleges to attract students with higher merit. Our main result shows that scholarships do not provide such a guarantee in general matching markets. For any school choice rule induced by school priority scores, there exists a market in which introducing a scholarship causes the school to be matched with a strictly worse set of students according to Gale-merit domination. Thus, scholarship provision can backfire even when scholarships are explicitly intended to attract higher-merit students. We then identify the precise condition under which this negative conclusion can be avoided. Under universal scholarship eligibility, a common priority ranking is necessary and sufficient for scholarship provision to guarantee a weak improvement in the scholarship-providing school's student pool for every student-preference profile satisfying within-school monotonicity. This condition is particularly relevant in centralized admission systems in which all schools rank students using the same scores from the entrance examination.

	\medskip
	\begin{center}
		\textit{Declaration of generative AI and AI-assisted technologies in the manuscript preparation process } 
	\end{center}

	\textit{During the preparation of this work, the author used ChatGPT and Claude for proof checking and writing polishing. The author reviewed and edited the output as needed and take full responsibility for the content of the published article.}

	\bibliography{ref_scholarship}
	\bibliographystyle{apalike}
	

	\newpage
	\begin{appendices}
    \section{Omitted Proofs}\label{sec:appen}
		\subsection{Main Results}
		\subsubsection*{Proof of Proposition~\ref{prop:characterization}}
		 \begin{proof}
			\noindent
			\emph{(If part.)}
			Let \(X\subseteq\mathcal X_s\). By construction, \(\ch_s^{Sc}(X)\) contains at most one contract for each student, contains at most \(q_s\) contracts, and contains at most \(q_s^m\) contracts at every positive scholarship level \(m\). Therefore, it satisfies \emph{scholarship feasibility}.
			
			Now consider a contract \(x\in X\setminus \ch_s^{Sc}(X)\). If student \(\iota(x)\) is selected at a scholarship level strictly higher than \(\mathbf m(x)\), then the conclusion required by \emph{scholarship maximality} holds whenever its premise is satisfied, while \emph{no justified envy} does not apply. Suppose instead that student \(\iota(x)\) is not selected at a higher level. Then she remains available when the algorithm processes level \(\mathbf m(x)\). If the effective capacity of that level were not exhausted, the algorithm would select \(x\). Hence, the effective capacity must be exhausted. Because the algorithm selects students in decreasing order of merit within each level, every student selected at level \(\mathbf m(x)\) has a higher merit score than \(\iota(x)\). Thus, \emph{no justified envy} is satisfied. Therefore, \(\ch_s^{Sc}\) satisfies all three axioms.
			
			\medskip
			
			\noindent
			\emph{(Only-if part.)}
			Suppose that a choice rule \(D\) satisfies \emph{scholarship feasibility}, \emph{scholarship maximality}, and \emph{no justified envy}. Fix an arbitrary set of contracts \(X\subseteq\mathcal X_s\), and write $Y:=D(X)$ and $C:=\ch_s^{Sc}(X).$ We show that \(Y=C\). 
			
			Order the positive scholarship levels as $m_K>m_{K-1}>\cdots>m_1>0.$ We proceed from the highest scholarship level to the lowest.
			
			For each \(k\), let \(H_k\) denote the set of students selected at scholarship levels strictly higher than \(m_k\). By induction, this set will be the same under \(Y\) and \(C\). Define
			\[
			A_k :=
			\left\{
			i\notin H_k: \text{ there exists }x\in X \text{ with }\iota(x)=i \text{ and }\mathbf m(x)=m_k
			\right\}.
			\]
			\(A_k\) is the set of students who have an available contract at level \(m_k\) and have not been selected at a higher level. We claim that both \(Y\) and \(C\) select at level \(m_k\) the highest-merit \(\min\{|A_k|,q_s^{m_k}\}\) students in \(A_k\).
			
			First, scholarship feasibility implies $|Y^{m_k}|\leq q_s^{m_k}$ and, because \(Y\) contains at most one contract for each student, every student selected in \(Y^{m_k}\) belongs to \(A_k\).
			
			Suppose for contradiction that $|Y^{m_k}|<\min\{|A_k|,q_s^{m_k}\}.$ Then there exists a student \(i\in A_k\) whose level-\(m_k\) contract \(x\) is not selected by \(Y\). Because $|Y^{m_k}|<q_s^{m_k}$, scholarship maximality implies that \(D(X)\) selects another contract \(z\) of student \(i\) with $\mathbf m(z)>m_k$. This contradicts \(i\in A_k\), since students in \(A_k\) have not been selected at a higher scholarship level. Therefore,
			\[
			|Y^{m_k}|=\min\{|A_k|,q_s^{m_k}\}.
			\]

			If \(|A_k|\leq q_s^{m_k}\), this equality implies that every student in \(A_k\) is selected at level \(m_k\). If \(|A_k|>q_s^{m_k}\), then $|Y^{m_k}|=q_s^{m_k}.$ For every rejected student \(i\in A_k\), no contract associated with \(i\) is selected at a higher scholarship level. No justified envy therefore implies that every student selected in \(Y^{m_k}\) has a higher merit score than \(i\). Hence, \(Y^{m_k}\) consists precisely of the \(q_s^{m_k}\) highest-merit students in \(A_k\). This is exactly the selection made by \(\ch_s^{Sc}\) at level \(m_k\). Thus, $Y^{m_k}=C^{m_k}.$ 
			
			At the highest level \(m_K\), the set \(H_K\) is empty, so the argument establishes equality at \(m_K\). Applying the same argument recursively establishes $Y^{m_k}=C^{m_k}$ for every positive scholarship level \(m_k\).
			
			It remains to consider the no-scholarship level. Let
			\[
			H_0 :=
			\left\{
			i: \text{student }i\text{ is selected at some positive scholarship level}
			\right\}.
			\]
			We have already shown that \(H_0\) is the same under \(Y\) and \(C\). Let
			\[
			A_0 :=
			\left\{
			i\notin H_0: \text{ there exists }x\in X \text{ with }\iota(x)=i \text{ and }\mathbf m(x)=0.
			\right\}
			\]
			Scholarship feasibility implies
			\[
			|Y^0|\leq \bar q_s^0(Y)=q_s-|H_0|.
			\]
			
			Suppose for contradiction that $|Y^0|<\min\{|A_0|,\bar q_s^0(Y)\}.$ Then there exists a student \(i\in A_0\) whose no-scholarship contract \(x\) is rejected. Because $|Y^0|<\bar q_s^0(Y)$, scholarship maximality implies that student \(i\) is selected under a positive-scholarship contract, contradicting \(i\in A_0\). Therefore,
			\[
			|Y^0|=\min\{|A_0|,\bar q_s^0(Y)\}.
			\]

			If \(|A_0|\leq\bar q_s^0(Y)\), every student in \(A_0\) is selected. If \(|A_0|>\bar q_s^0(Y)\), no positive-scholarship contract associated with a student in \(A_0\) is selected. No justified envy therefore implies that the \(\bar q_s^0(Y)\) selected students are precisely the highest-merit students in \(A_0\). Since \(Y\) and \(C\) coincide at every positive scholarship level, they have the same residual capacity; hence, \(C\) makes exactly the same selection at level zero. Thus, $Y^0=C^0.$

			We have shown that \(Y\) and \(C\) coincide at every positive scholarship level and at the no-scholarship level. Hence, $D(X) = \ch^{Sc}_s(X)$. Since \(X\subseteq\mathcal X_s\) was chosen arbitrarily, \(D=\ch_s^{Sc}\).
		\end{proof}

		\subsubsection*{Proof of Proposition~\ref{prop:PathI}}

		\begin{proof}
			We first show that $\ch^{Sc}_s$ is not path independent. Write
			$x_i^m=(i,s,m)$, let
			$\pi_s(i_1)>\cdots>\pi_s(i_5)$, and set
			$q_s^{m_2}=q_s^{m_1}=1$ and $q_s=3$, where $m_2>m_1>0$.
			Consider
			\[
			X=\{x_1^{m_2},x_1^0,x_2^{m_1},x_3^{m_2},x_3^{m_1},
			x_4^{m_1},x_4^0,x_5^{m_2},x_5^{m_1},x_5^0\}.
			\]
			The original rule chooses $\ch^{Sc}_s(X)=\{x_1^{m_2},x_2^{m_1},x_4^0\}.$ Let
			\[
			Y=\{x_1^0,x_2^{m_1},x_3^{m_2},x_4^0\},
			\qquad Z=X\setminus Y.
			\]
			Then $\ch^{Sc}_s(Y)=\{x_1^0,x_2^{m_1},x_3^{m_2}\},$ and a direct application of the rule yields
			\[
			\ch^{Sc}_s\bigl(Z\cup \ch^{Sc}_s(Y)\bigr)
			=\{x_1^{m_2},x_2^{m_1},x_5^0\}
			\neq \ch^{Sc}_s(Z\cup Y)=\ch^{Sc}_s(X).
			\]
			Thus $\ch^{Sc}_s$ is not path independent. It remains to construct the claimed path-independent and size-monotonic completion.

			We propose a completion. The completion treats each positive scholarship level as a separate competition. A student selected at one level remains eligible at every other level. Thus, selection at a higher scholarship level neither guarantees nor prevents selection of the same student's lower-level contract. After making all positive-level choices, the rule fills the remaining total capacity with the highest-merit zero-scholarship contracts. Because the positive scholarship level choices are independent, we can process in any order of scholarship level $m \in M^s$.

			\setcounter{algocf}{1}
			\begin{algorithm}[H]
			\caption{Level-by-Level Scholarship Choice Completion $\overline{\ch}^{Sc}_s$}
			\label{alg:level-completion}
			\KwIn{A set of contracts $X\subseteq\mathcal X_s$}
			\KwOut{The chosen set $\overline{\ch}^{Sc}_s(X)$}

			$C^+\leftarrow\varnothing$\;
			\ForEach{$m\in M^s$, in any order}{
			$X^m\leftarrow\{x\in X:\mathbf m(x)=m\}$\;
			$C^m\leftarrow$ the top $\min\{|X^m|,q_s^m\}$ contracts in $X^m$ according to $\pi_s$\;
			$C^+\leftarrow C^+\cup C^m$\;
			}
			$r\leftarrow q_s-|C^+|$\;
			$X^0\leftarrow\{x\in X:\mathbf m(x)=0\}$\;
			$C^0\leftarrow$ the top $\min\{|X^0|,r\}$ contracts in $X^0$ according to $\pi_s$\;
			\Return{$C^+\cup C^0$}\;
			\end{algorithm}

			Every chosen contract counts separately toward total capacity, even when several chosen contracts are associated with the same student. Since $\sum_{m\in M^s}q_s^m\leq q_s$, the residual capacity $r$ is always nonnegative. Therefore, Algorithm~\ref{alg:level-completion} respects both total capacity and every positive scholarship capacity.

			For each $m\in M^s$, let $T_m(X)$ be the top $\min\{|X^m|,q_s^m\}$ contracts in $X^m$ according to $\pi_s$, and define
			\[
			a(X):=\sum_{m\in M^s}|T_m(X)|
			=\sum_{m\in M^s}\min\{|X^m|,q_s^m\},
			\qquad r(X):=q_s-a(X).
			\]
			For any menu $Z\subseteq\mathcal X_s$ and nonnegative integer $k$, let $T_0^k(Z)$ denote the top $\min\{|Z^0|,k\}$ contracts in $Z^0$. Algorithm~\ref{alg:level-completion} can then be written as
			\begin{equation}\label{eq:level-representation}
			\overline{\ch}^{Sc}_s(X)
			=T_0^{r(X)}(X) \cup \bigcup_{m\in M^s}T_m(X).
			\end{equation}

			Next, we propose three lemmas to complete the proof. The first shows that $\overline{\ch}^{Sc}_s$ is indeed a completion of $\ch^{Sc}_s$, the second and the third verify that $\overline{\ch}^{Sc}_s$ satisfies size monotonicity and path independence, respectively. 

			\begin{lemma}\label{lem:completion}
			The choice rule $\overline{\ch}^{Sc}_s$ is a completion of $\ch^{Sc}_s$.
			\end{lemma}
			\begin{proof}
			Fix $X\subseteq\mathcal X_s$. If $\overline{\ch}^{Sc}_s(X)$ contains two contracts associated with the same student, then the first condition in Definition~\ref{def:completion} holds.

			Suppose instead that $\overline{\ch}^{Sc}_s(X)$ contains at most one
			contract for each student. We use the following elementary observation:
			deleting contracts that lie outside the top $k$ contracts of a strictly
			ordered set does not change its top-$k$ choice. At the highest positive
			level, $\overline{\ch}^{Sc}_s$ and $\ch^{Sc}_s$ make the same choice.
			Proceeding downward, suppose the rules have agreed at all higher levels.
			The original rule deletes the lower-level contracts of students already
			selected. None of those deleted contracts belongs to $T_m(X)$; otherwise
			$\overline{\ch}^{Sc}_s(X)$ would contain two contracts for that student.
			The observation therefore implies that the two rules also agree at level
			$m$. Induction gives agreement at every positive level.

			The two rules consequently have the same residual total capacity. Any
			zero-scholarship contract deleted by the original rule belongs to a student
			selected at a positive level and, by the maintained supposition, lies
			outside $T_0^{r(X)}(X)$. The same observation gives agreement at level
			zero. Hence $\overline{\ch}^{Sc}_s(X)=\ch^{Sc}_s(X)$, so the second
			condition in Definition~\ref{def:completion} holds.
			\end{proof}

			\begin{lemma}\label{lem:size-monotonic}
			The choice rule $\overline{\ch}^{Sc}_s$ is size monotonic.
			\end{lemma}
			\begin{proof}
			Equation~\eqref{eq:level-representation} gives
			\begin{align*}
			|\overline{\ch}^{Sc}_s(X)|
			&=a(X)+\min\{|X^0|,q_s-a(X)\}\\
			&=\min\left\{q_s,
			|X^0|+\sum_{m\in M^s}\min\{|X^m|,q_s^m\}\right\}.
			\end{align*}
			If $X\subseteq Y$, each term inside the final minimum weakly increases.
			Thus $|\overline{\ch}^{Sc}_s(X)|\leq
			|\overline{\ch}^{Sc}_s(Y)|$, as required by
			Definition~\ref{def:sizemon}.
			\end{proof}

			\begin{lemma}\label{lem:completion-pi}
			The choice rule $\overline{\ch}^{Sc}_s$ is path independent.
			\end{lemma}
			\begin{proof}
			Fix $X,Y\subseteq\mathcal X_s$. At every positive level $m$, the
			standard top-$q_s^m$ identity gives
			\[
			T_m(X\cup Y)=T_m\bigl(X\cup T_m(Y)\bigr).
			\]
			Hence the two menus $X\cup Y$ and
			$X\cup\overline{\ch}^{Sc}_s(Y)$ generate the same positive-level
			choices and therefore the same residual capacity. Denote it by
			\[
			r^*:=r(X\cup Y)
			=r\bigl(X\cup\overline{\ch}^{Sc}_s(Y)\bigr).
			\]
			Moreover, $r^*\leq r(Y)$ because adding $X$ can only weakly increase
			the number of positive-level contracts selected.

			Let $B:=T_0^{r(Y)}(Y)$. If $B=Y^0$, no zero-scholarship contract is
			removed from $Y$. Otherwise, $|B|=r(Y)$, and every contract in
			$Y^0\setminus B$ is ranked below every contract in $B$. Since
			$r^*\leq r(Y)$, no removed contract can enter the top $r^*$ after the
			contracts in $X$ are added. In either case,
			\[
			T_0^{r^*}(X\cup Y)
			=T_0^{r^*}\bigl(X\cup\overline{\ch}^{Sc}_s(Y)\bigr).
			\]
			Combining the positive- and zero-level choices proves
			\[
			\overline{\ch}^{Sc}_s(X\cup Y)
			=\overline{\ch}^{Sc}_s
			\bigl(X\cup\overline{\ch}^{Sc}_s(Y)\bigr).
			\]
			\end{proof}

			Lemmas~\ref{lem:completion},
			\ref{lem:size-monotonic}, and~\ref{lem:completion-pi} establish that
			$\overline{\ch}^{Sc}_s$ is the claimed completion.
		\end{proof}
		
		\subsubsection*{Proof of Proposition~\ref{prop:characterizeMech}}
		\begin{proof}
			We abbreviate the student-proposing DA mechanism based on $(\ch^{Sc}_s)_{s \in \mathcal{S}}$ as $DA^{Sc}$.

			\medskip
			\noindent
			\emph{(If part.)} Assume $\varphi(\succeq_{\mathcal{I}})$ is outcome equivalent to $DA^{Sc}(\succeq_{\mathcal{I}})$ for any $\succeq_{\mathcal{I}} \in \mathcal{R}^{WM}$. By \citet{hatfield2015hidden}, we know that running the cumulative-offer mechanism based on $(\ch^{Sc}_s)_{s \in \mathcal{S}}$ is equivalent to running the cumulative-offer mechanism based on $(\overline{\ch}^{Sc}_s)_{s \in \mathcal{S}}$ since Proposition~\ref{prop:PathI} shows it is a path-independent and size-monotonic completion. Thus, by \citet[Theorems A.2 and A.3]{hatfield2015hidden}, we know that the cumulative-offer mechanism based on $(\ch^{Sc}_s)_{s \in \mathcal{S}}$ is strategy-proof and produces a stable outcome. Furthermore, \citet{Hatfield2021Stability} showed that choice rule completion implies observable substitutability. From \citet[Proposition A.1]{Hatfield2021Stability}, we know that $DA^{Sc}(\succeq_{\mathcal{I}})$ is equivalent to the cumulative-offer mechanism based on $(\ch^{Sc}_s)_{s \in \mathcal{S}}$ given student-preference profile $\succeq_{\mathcal{I}}$. Therefore, we obtain that the matching outcome, denoted by $\mu^{Sc}$, generated from $DA^{Sc}(\succeq_{\mathcal{I}})$ is stable and strategy-proof. Thus, individual rationality and strategy-proofness are satisfied. 
			
			Next, consider the demand set $D_s(\mu^{Sc}, \succeq_{\mathcal I})$ for a school $s \in \mathcal{S}$. Lemma~\ref{lemma:matchequalchoice} shows that $\mu^{Sc}_s = \ch^{Sc}_s(D_s(\mu^{Sc}, \succeq_{\mathcal I}))$ for each school $s \in \mathcal{S}$, and by Proposition~\ref{prop:characterization}, we get \emph{scholarship feasibility}, \emph{scholarship maximality}, and \emph{no justified envy} for all $s \in \mathcal{S}$. Therefore, $DA^{Sc}(\succeq_{\mathcal{I}})$ satisfies the extension of \emph{scholarship feasibility}, \emph{scholarship maximality}, and \emph{no justified envy}.
			
			\medskip
			\noindent
			\emph{(Only-if part.)} Let \(\varphi\) be a mechanism satisfying properties (1)--(5) on \(\mathcal R^{WM}\). Fix \(\succeq_{\mathcal I}\in\mathcal R^{WM}\), and let $\mu:=\varphi(\succeq_{\mathcal I}).$ We will show that $\mu = \varphi(\succeq_{\mathcal{I}}) = DA^{Sc}(\succeq_{\mathcal{I}})$. Since $\mu_s \in \psi^F_s \bigl(D_s(\mu, \succeq_{\mathcal{I}})\bigr) \cap \psi^M_s \bigl(D_s(\mu, \succeq_{\mathcal{I}})\bigr) \cap \psi^N_s \bigl(D_s(\mu, \succeq_{\mathcal{I}})\bigr) = \{\ch^{Sc}_s(D_s(\mu, \succeq_{\mathcal{I}}))\}$, we have $\mu_s = \ch^{Sc}_s \bigl(D_s(\mu, \succeq_{\mathcal{I}})\bigr)$ for all $s \in \mathcal{S}$. 
			
			Next, we prove that $\mu$ is stable. By Lemma~\ref{lemma:demandcompletioneqm}, $\ch_s^{Sc}\bigl(D_s(\mu,\succeq_{\mathcal I})\bigr) = \overline{\ch}_s^{Sc}\bigl(D_s(\mu,\succeq_{\mathcal I})\bigr) = \mu_s.$ Because \(\overline{\ch}_s^{Sc}\) is path independent, it is consistent. Hence, for every \(Y\) such that $\mu_s\subseteq Y \subseteq D_s(\mu,\succeq_{\mathcal I}),$ we have
			\[
			\overline{\ch}_s^{Sc}(Y)=\mu_s.
			\]
			Since \(\mu_s\) contains at most one contract associated with each student, the definition of completion implies $\ch_s^{Sc}(Y) = \overline{\ch}_s^{Sc}(Y) = \mu_s.$ Taking \(Y=\mu_s\) gives school individual rationality under both choice rules. 
			
			Now suppose that \(Z\) is a candidate blocking set at school \(s\). The student-side blocking condition implies $Z\subseteq D_s(\mu,\succeq_{\mathcal I}),$ and therefore $\mu_s \subseteq \mu_s \cup Z \subseteq D_s(\mu,\succeq_{\mathcal I}).$ It follows that
			\[
			\ch_s^{Sc}(\mu_s\cup Z) = \overline{\ch}_s^{Sc}(\mu_s\cup Z) = \mu_s,
			\]
			hence \(Z\) cannot block under either choice rule. Together with student individual rationality, \(\mu\) is stable with respect to both $\ch_s^{Sc}$ and $\overline{\ch}_s^{Sc}$.

			Because $\succeq_{\mathcal{I}} \in \mathcal{R}^{WM}$ was arbitrary, we have now shown that $\varphi$ is a stable mechanism on $\mathcal{R}^{WM}$. By assumption, it is strategy-proof on this domain. Lemma~\ref{lem:sg-irc} and Lemma~\ref{lem:wm-admissible-subclass} verify the remaining conditions needed to apply the restricted-domain version of the characterization from \citet{Hatfield2021Stability}.

			By Proposition~\ref{prop:PathI}, \(\ch_s^{Sc}\) has a path-independent and size-monotonic completion. Together with Lemma~\ref{lem:sg-irc}, the result from \citet{Hatfield2021Stability} implies that \(\ch_s^{Sc}\) is observably substitutable on the full preference domain, and therefore also on \(\mathcal R^{WM}\).

			Let \(COM^{Sc}\) denote the cumulative-offer mechanism based on the choice-rule profile \(\bigl(\ch_s^{Sc}\bigr)_{s\in\mathcal S}\). Appendix B.4 of \citet{Hatfield2021Stability} explains that Theorem 1b continues to hold when student preferences are restricted to an admissible subclass. Because \(\varphi\) is stable and strategy-proof on \(\mathcal R^{WM}\), the restricted-domain version of Theorem 1b implies that, for every \(\succeq_{\mathcal I}\in\mathcal R^{WM}\),
			\[
			\varphi(\succeq_{\mathcal I})
			=
			COM^{Sc}(\succeq_{\mathcal I}).
			\]
			
			Finally, Proposition A.1 of \citet{Hatfield2021Stability} implies that, under observable substitutability, the cumulative-offer and deferred-acceptance mechanisms are outcome equivalent. Hence,
			\[
			\varphi(\succeq_{\mathcal I})
			=
			COM^{Sc}(\succeq_{\mathcal I})
			=
			DA^{Sc}(\succeq_{\mathcal I}),
			\]
			as desired.
		\end{proof}

		\subsubsection*{Proof of Theorem~\ref{thm:neg}}
		\begin{proof}
			We construct the following market:
			
			Let there be two schools $\{s_1,s_2\}$ with $q_{s_1} = q_{s_2} = 1$ and two students $\{i_1,i_2\}$. School $s_2$ introduces scholarship level $m$ with capacity $q^m_{s_2}=1$, thereby adding the contracts $(i_1,s_2,m)$ and $(i_2,s_2,m)$. The augmented preferences of the students and the priorities of the schools are as follows:
			\begin{table}[H]
				\centering
				\begin{minipage}{0.4\textwidth}
					\centering
					\begin{tabular}{c|c}
						$\succ^A_{i_1}$ & $\succ^A_{i_2}$ \\
						\hline
						$s_1$ & $(s_2,m)$ \\
						$(s_2,m)$ & $s_1$ \\
						$s_2$ & $s_2$ \\
					\end{tabular}      
				\end{minipage}        
				\begin{minipage}{0.4\textwidth}
					\centering
					\begin{tabular}{c|c}
						$\pi_{s_1}$ & $\pi_{s_2}$\\
						\hline
						$i_2$ & $i_1$ \\
						$i_1$ & $i_2$ \\ 
					\end{tabular}            
				\end{minipage}    
			\end{table}
			
			Before school $s_2$ introduces scholarship level $m$, both students prefer $s_1$ to $s_2$. Therefore, the matching from the student-proposing DA mechanism is:
			\begin{center}
				$ \mu = 
				\begin{pmatrix}
					s_1 & s_2 \\
					i_2& i_1 \\
				\end{pmatrix}
				$
			\end{center}
			
			After $s_2$ introduces scholarship level $m$, student $i_2$ proposes the scholarship contract to $s_2$. This creates a vacancy at $s_1$, and $i_1$ matches with $s_1$ instead of $s_2$. Therefore, the stable matching generated by the student-proposing DA mechanism is:
			\begin{center}
				$ \mu' = 
				\begin{pmatrix}
					s_1 & s_2 \\
					i_1 & (i_2,m) \\
				\end{pmatrix}
				$
			\end{center}
			Thus, $\iota(\mu_{s_2})$ strictly Gale-dominates $\iota(\mu'_{s_2})$. Notice that this works for any choice rule that is induced by the priority $\pi_{s_i}$, $i = 1, 2$.
		\end{proof}

		\subsubsection*{Proof of Theorem~\ref{thm:commonrank}}
		
		\begin{proof}
			Write the common merit order as $i_1,i_2,\ldots,i_n$, where $i_1$ has the highest merit. A contract $x$ for student $i_r$ is \emph{available at her turn} if adding $x$ to the assignments of $i_1,\ldots,i_{r-1}$ satisfies the total and scholarship-level capacities. The merit serial dictatorship calls students in this order and assigns each student her most-preferred acceptable contract available at her turn.

			We use two claims. The first connects this serial dictatorship to the mechanism in the theorem.

			\begin{claim}\label{claim:sd-common}
			Under common ranking and universal scholarship eligibility, the student-proposing DA outcome based on $(\ch_s^{Sc})_{s\in\mathcal S}$ is the merit-serial-dictatorship outcome.
			\end{claim}
			\begin{proof}
			Let $\nu$ be the serial-dictatorship outcome. We first show that $\nu$ is stable. It is individually rational by construction, and $\ch_s^{Sc}(\nu_s)=\nu_s$ because $\nu_s$ is feasible and contains at most one contract per student.

			We first record an implication of universal scholarship eligibility and within-school monotonicity. Let $H_s^r$ be the contracts assigned to $s$ before student $i_r$ is called, and let $H_s^{r,\ell}:=\{x\in H_s^r:\mathbf m(x)=\ell\}$. For every contract $(i_r,s,\ell)$,
			\[
			(i_r,s,\ell)\text{ is unavailable at }i_r\text{'s turn}
			\quad\Longrightarrow\quad
			\begin{cases}
			|H_s^r|=q_s, & \ell=0,\\
			|H_s^{r,\ell}|=q_s^\ell, & \ell>0.
			\end{cases}
			\tag{$*$}
			\]
			Only the positive-level implication requires explanation. If the total capacity were filled while the level-$\ell$ capacity were not, then, because $\sum_{h\in M^s}q_s^h\leq q_s$, at least one student in $H_s^r$ would hold a zero-scholarship contract. When that higher-merit student was called, her level-$\ell$ contract was available and preferred to her zero-scholarship contract, a contradiction. In particular, if a student is assigned a zero-scholarship contract at $s$, every positive scholarship level at $s$ is already filled by higher-merit students when she is called.

			Suppose that a nonempty set $Z$ blocks $\nu$, and let $i_r$ be the highest-merit student represented in $Z$. Write $x=(i_r,s,\ell)$ for her contract in $Z$. Because $i_r$ prefers $x$ to $\nu_{i_r}$, contract $x$ was not available at her turn. By $(*)$, contracts of students ranked above $i_r$ fill the capacity that prevents $x$ from being selected. If $\ell=0$, every positive-level contract at $s$ is preferred by $i_r$ to $x$ and hence was also unavailable at her turn. By $(*)$, every positive level is therefore filled by higher-merit students. Consequently, lower-merit members of $Z$ cannot enter a positive level and displace a higher-merit zero-level student. Because $i_r$ is the highest-merit member of $Z$, adding $Z$ cannot remove the obstruction to $x$. Hence $x\notin\ch_s^{Sc}(\nu_s\cup Z)$, contradicting that $Z$ blocks $\nu$.

			Next, let $\widehat\nu$ be any stable matching. We show inductively that it equals $\nu$. Suppose the assignments of $i_1,\ldots,i_{r-1}$ agree, and let $x=(i_r,s,\ell)$ be the contract chosen by $i_r$ in the serial dictatorship. Under $\widehat\nu$, student $i_r$ cannot receive a contract that she prefers to $x$. Since $\widehat\nu$ agrees with $\nu$ on the assignments of all higher-merit students, any such contract assigned to $i_r$ under $\widehat\nu$ would have been available at her serial-dictatorship turn, contradicting the definition of $x$ as her most-preferred available contract.

			If $i_r$ receives a worse contract, then $x$ blocks $\widehat\nu$. Every level at $s$ that $i_r$ prefers to $\ell$ was unavailable at her turn and hence, by $(*)$, is already filled by higher-merit students. If $\ell>0$, contract $x$ is available at level $\ell$ after those higher-merit assignments and outranks every later student at that level. If $\ell=0$, every positive level at $s$ is filled by higher-merit students, so $\widehat\nu_s$ cannot contain a lower-merit student at a positive level; among the residual zero-level contracts, $x$ is available and outranks every later student. Therefore, in either case, $x\in\ch_s^{Sc}(\widehat\nu_s\cup\{x\})$, contradicting stability. It follows that $\widehat\nu_{i_r}=x$. This proves uniqueness. The DA outcome is stable by Proposition~\ref{prop:PathI} and the completion results used above, so it must equal $\nu$.
			\end{proof}

			The second claim records the comparative-static property of serial dictatorship that we need.

			\begin{claim}[Vacancy-chain claim]\label{claim:vacancy-chain}
			Fix capacities and a serial-dictatorship order. Suppose that level-$m$ contracts at $s_k$ were already introduced for every student ranked above $i_r$, but not for $i_r$ or any student ranked below $i_r$. Let $\nu^-$ and $\nu^+$ be the serial-dictatorship outcomes before and after $(i_r,s_k,m)$ is added, respectively. Then
			\[
			\iota(\nu^+_{s_k})\trianglerighteq^G\iota(\nu^-_{s_k}).
			\]
			\end{claim}
			\begin{proof}
			For each $t$, let $\nu^{-,t}$ and $\nu^{+,t}$ be the partial outcomes after students $i_1,\ldots,i_t$ have been called in the two runs, and define the occupancy gap
			\[
			\Delta_t
			:=
			|\iota(\nu^{+,t}_{s_k})|
			-
			|\iota(\nu^{-,t}_{s_k})|.
			\]
			The assignments of $i_1,\ldots,i_{r-1}$ coincide, so $\Delta_t=0$ for $t<r$. If $i_r$ does not choose the new contract, the two runs remain identical. Suppose that she chooses it. If she was not previously assigned to $s_k$, then $\Delta_r=1$; if she was already assigned to $s_k$, then $\Delta_r=0$.

			We now carry the vacancy created by this move down the merit order. Whenever a later student's assignment differs, either her assignment uses capacity released by an earlier difference, or her former assignment becomes unavailable because capacity was consumed by an earlier difference. After she moves, her former assignment becomes the next vacancy. Thus, possibly with unchanged students between successive links, every difference belongs to the vacancy chain initiated by $i_r$, and students appear along the chain in decreasing merit order.

			We prove inductively that
			\[
			\Delta_t\in\{0,1\}
			\tag{$**$}
			\]
			and that crossings of the chain into and out of $s_k$ alternate, beginning with an entry. The assertion holds immediately after $i_r$ is called. A link that does not cross the boundary of $s_k$ leaves $\Delta_t$ unchanged. An entry into $s_k$ consumes the vacancy at $s_k$ created by the preceding departure; except for the initial move of $i_r$, it therefore occurs only when the gap is zero and changes it to one. A departure from $s_k$ either responds to the additional unit of total capacity already used in the new run or carries the vacancy chain to a contract made available outside $s_k$. In either case, the chain must previously have entered $s_k$ without a subsequent departure. Hence a departure occurs only when the gap is one and changes it to zero. No student ranked below $i_r$ has a level-$m$ contract, so use of the new level cannot create an additional departure through its level capacity. Changes between pre-existing levels at $s_k$ do not cross the school's boundary and merely move the vacancy between levels. This proves $(**)$ and the alternation assertion.

			Pair each departure from $s_k$ with the immediately preceding unpaired entry. The alternation invariant makes these pairs distinct, and the order of the chain implies that the entrant has strictly higher merit than the departing student. Every unpaired change at $s_k$ is an entry. Extend this pairing to an injection from $\iota(\nu^-_{s_k})$ to $\iota(\nu^+_{s_k})$ by mapping students assigned to $s_k$ in both runs to themselves. Lemma~\ref{lem:gale-injection} then gives
			\[
			\iota(\nu^+_{s_k})\trianglerighteq^G\iota(\nu^-_{s_k}).
			\]
			\end{proof}

			\noindent
			\emph{(If part.)}
			Introduce the new scholarship contracts one at a time in common-merit order. Formally, let $\mu^0=\mu$, and, for each $r=1,\ldots,n$, let $\mu^r$ be the merit-serial-dictatorship outcome after the contracts
			\[
			(i_1,s_k,m),\ldots,(i_r,s_k,m)
			\]
			have been added. Assigning capacity $q_{s_k}^m$ to level $m$ before any of these contracts is added does not affect the outcome. Claim~\ref{claim:vacancy-chain} gives
			\[
			\iota(\mu^r_{s_k})\trianglerighteq^G
			\iota(\mu^{r-1}_{s_k})
			\qquad\text{for every }r.
			\]
			By Lemma~\ref{lem:gale-transitive},
			\[
			\iota(\mu^n_{s_k})\trianglerighteq^G
			\iota(\mu^0_{s_k}).
			\]
			By universal scholarship eligibility, $\mu^n=\mu'$, and Claim~\ref{claim:sd-common} identifies $\mu^0$ and $\mu^n$ with the two DA outcomes. Therefore,
			\[
			\iota(\mu'_{s_k})\trianglerighteq^G\iota(\mu_{s_k}).
			\]

			\medskip

			\noindent
			\emph{(Only-if part.)}
			We prove the contrapositive. If the priority profile does not have a common ranking, there are students $i,j$ and schools $s_a,s_b$ such that
			\[
			\pi_{s_a}(j)>\pi_{s_a}(i)
			\quad\text{and}\quad
			\pi_{s_b}(i)>\pi_{s_b}(j).
			\]
			Construct a market with $q_{s_a}=q_{s_b}=1$. School $s_a$ initially offers level $m$ with $q^m_{s_a}=1$, while $s_b$ introduces the same level with $q^m_{s_b}=1$. The new level is available to every student. Give $i$ and $j$ the augmented preferences
			\begin{align*}
			 i:\quad &(i,s_a,m)\succ_i(i,s_b,m)\succ_i(i,s_b,0)
			       \succ_i(i,s_a,0)\succ_i\emptyset,\\
			 j:\quad &(j,s_b,m)\succ_j(j,s_a,m)\succ_j(j,s_a,0)
			       \succ_j(j,s_b,0)\succ_j\emptyset.
			\end{align*}
			Place every omitted contract below the outside option, ordered consistently with within-school monotonicity, and do the same for all other students so that these contracts do not affect the construction. Before $s_b$ introduces level $m$, school $s_a$ admits $j$ and $s_b$ admits $i$. After the introduction, $s_a$ admits $i$ and $s_b$ admits $j$. Consequently,
			\[
			\iota(\mu_{s_b})=\{i\}\triangleright^G
			\{j\}=\iota(\mu'_{s_b}),
			\]
			because $s_b$ ranks $i$ above $j$. Thus, a priority profile without common ranking does not provide the guarantee in Theorem~\ref{thm:commonrank}.
		\end{proof}

		\subsection{Technical Lemmas}

		\begin{lemma}\label{lem:gale-injection}
		Let $A$ and $B$ be sets of students ranked by $\pi_s$. If there is an injection $g:B\to A$ such that
		\[
		\pi_s(g(i))\geq\pi_s(i)
		\qquad\text{for every }i\in B,
		\]
		then $A\trianglerighteq^G B$.
		\end{lemma}
		\begin{proof}
		The injection implies $|A|\geq|B|$. List both sets in decreasing merit order. If the $j$-th student in $A$ had lower merit than the $j$-th student in $B$, the first $j$ students in $B$ would be mapped to $j$ distinct students in $A$ whose merit is at least that of the $j$-th student in $B$, a contradiction. Therefore, the $j$-th merit in $A$ is weakly higher than the $j$-th merit in $B$ for every $j\leq|B|$.
		\end{proof}

		\begin{lemma}\label{lem:gale-transitive}
		Gale domination is transitive.
		\end{lemma}
		\begin{proof}
		Suppose $A\trianglerighteq^G B$ and $B\trianglerighteq^G C$, and list each set in decreasing merit order. Then $|A|\geq|B|\geq|C|$ and, for every $j\leq|C|$,
		\[
		\pi_s(a_j)\geq\pi_s(b_j)\geq\pi_s(c_j).
		\]
		Hence $A\trianglerighteq^G C$.
		\end{proof}

		\begin{lemma}\label{lemma:matchequalchoice}
			For every $s\in \mathcal{S}$, $\mu^{Sc}_s = \ch^{Sc}_s(D_s(\mu^{Sc}, \succeq_{\mathcal I})).$
		\end{lemma}
		\begin{proof}
    		Let \(P_s^t\) denote the set of contracts proposed to school \(s\) through step \(t\) of the cumulative-offer process based on \((\ch_s^{Sc})_{s\in\mathcal S}\). The proof of Theorem A.2 in \citet{hatfield2015hidden} establishes that, at every step \(t\),
    		\[
    		\ch_s^{Sc}(P_s^t)
    		=
    		\overline{\ch}_s^{Sc}(P_s^t).
    		\]
    		Because the proposal menus expand over time and \(\overline{\ch}_s^{Sc}\) is substitutable, a contract rejected from \(P_s^t\) by \(\overline{\ch}_s^{Sc}\) cannot be chosen from any later menu \(P_s^{t'}\), where \(t'\geq t\). The step-by-step equality then implies the same conclusion for \(\ch_s^{Sc}\). Hence, every rejection in the cumulative-offer process based on \((\ch_s^{Sc})_{s\in\mathcal S}\) is final.

    		Let \(P_s(\succeq_{\mathcal I})\) denote the set of all contracts proposed to school \(s\) during this process. Because \(\mu^{Sc}\) is its terminal outcome,
    		\[
        	\mu^{Sc}_s = \ch^{Sc}_s\bigl(P_s(\succeq_{\mathcal I})\bigr).
    		\]
			It therefore suffices to show that $P_s(\succeq_{\mathcal I}) = D_s(\mu^{Sc},\succeq_{\mathcal I}).$ Fix a contract \(x \in \mathcal X_s\). Suppose first that \(\mu^{Sc}_{\iota(x)} \neq \emptyset\). Every contract that student \(\iota(x)\) strictly prefers to \(\mu^{Sc}_{\iota(x)}\) is proposed before she proposes \(\mu^{Sc}_{\iota(x)}\), and \(\mu^{Sc}_{\iota(x)}\) itself is also proposed. Furthermore, \(\iota(x)\) never proposes a contract ranked below $\mu^{Sc}_{\iota(x)}$. Therefore,
    		\[
        	x\in P_s(\succeq_{\mathcal I}) \quad\Longleftrightarrow\quad x\succeq_{\iota(x)}\mu^{Sc}_{\iota(x)}.
    		\]

    		If \(\mu^{Sc}_{\iota(x)}=\emptyset\), then student \(\iota(x)\) is unmatched only after exhausting all contracts she finds acceptable. Thus,
    		\[
        	x\in P_s(\succeq_{\mathcal I}) \quad \Longleftrightarrow \quad x \succ_{\iota(x)} \emptyset \quad \Longleftrightarrow \quad x \succeq_{\iota(x)} \mu^{Sc}_{\iota(x)}.
    		\]
    		The last equivalence uses strict preferences and \(x\neq\emptyset\). Therefore, $P_s(\succeq_{\mathcal I}) = \{x\in\mathcal X_s: x\succeq_{\iota(x)}\mu^{Sc}_{\iota(x)}\} = D_s(\mu^{Sc},\succeq_{\mathcal I}).$

			Substituting this identity gives $\mu^{Sc}_s = \ch^{Sc}_s\bigl(D_s(\mu^{Sc},\succeq_{\mathcal I})\bigr),$ as desired.
		\end{proof}

		\begin{lemma}\label{lemma:demandcompletioneqm}
			Let $\succeq_{\mathcal{I}} \in \mathcal{R}^{WM}$, and let $\mu$ satisfy $\mu_s = \ch^{Sc}_s(D_s(\mu, \succeq_{\mathcal{I}}))$ for all $s \in \mathcal{S}$. Then, $$\mu_s = \ch^{Sc}_s(D_s(\mu, \succeq_{\mathcal{I}})) = \overline{\ch}^{Sc}_s(D_s(\mu, \succeq_{\mathcal{I}})).$$
		\end{lemma}
		\begin{proof}
			Let $D:=D_s(\mu,\succeq_{\mathcal I}).$ At the highest scholarship level, \(\ch_s^{Sc}\) and \(\overline{\ch}_s^{Sc}\) select the same highest-merit contracts. Every contract selected by \(\ch_s^{Sc}\) at that level belongs to \(\mu_s\). If \(x_i^m=(i,s,m) \in \mu_s\), then \(\mu_i=x_i^m\). By within-school monotonicity, every contract \(x_i^\ell=(i,s,\ell)\) with \(\ell<m\) is strictly worse than \(\mu_i\). Therefore, $x_i^\ell\notin D.$ Thus, none of the students selected at the highest level has a lower-level contract in \(D\).
			
			Proceed inductively. Suppose the two rules have made the same selections at every level above \(m\). The original rule removes the lower-level contracts of students selected at those higher levels. All such contracts are already absent from \(D\). Hence, at level \(m\), the two rules face the same set of contracts and make the same merit-based selection. The same argument applies recursively at every positive scholarship level and at the zero-scholarship level. Therefore,
			\[
    		\overline{\ch}_s^{Sc}(D) = \ch_s^{Sc}(D) =	\mu_s.
			\]
		\end{proof}

		\begin{lemma}\label{lem:sg-irc}
			\(\ch_s^{Sc}\) satisfies irrelevance of rejected contracts.
		\end{lemma}
		\begin{proof}
		Fix \(X\subseteq\mathcal X_s\) and \(x\in X\setminus\ch_s^{Sc}(X)\). Let $Y:=X\setminus\{x\}.$ We show that $\ch_s^{Sc}(Y)=\ch_s^{Sc}(X).$
		Write \(m=\mathbf m(x)\) and \(i=\iota(x)\). The presence or absence of \(x\) cannot affect the choices made at scholarship levels strictly above \(m\), so the same students are selected at those levels under \(X\) and \(Y\).

		If student \(i\) is selected at a level above \(m\), then the algorithm removes \(x\) before processing level \(m\). Thus, deleting \(x\) from the initial menu has no effect on any choice. Otherwise, \(x\) remains available when level \(m\) is processed. Because \(x\notin\ch_s^{Sc}(X)\), the capacity at level \(m\) must be filled by contracts associated with students having higher merit than \(i\). Removing \(x\) therefore does not change the contracts selected at level \(m\). Consequently, the same students and their lower-level contracts are removed under \(X\) and \(Y\), so all subsequent scholarship-level choices are also identical.

		The same argument applies when \(m=0\), using the residual total capacity after the positive-scholarship choices have been made. This proves the result for the deletion of one rejected contract. 
		
		Now let $\ch_s^{Sc}(X)\subseteq Y\subseteq X.$ Since the contract set is finite, remove the contracts in \(X\setminus Y\) one at a time. At every deletion the chosen set remains \(\ch_s^{Sc}(X)\), so every contract deleted at a subsequent step remains rejected. Hence
		\[
		\ch_s^{Sc}(Y)=\ch_s^{Sc}(X).
		\]
		Therefore, \(\ch_s^{Sc}\) satisfies irrelevance of rejected contracts.
		\end{proof}

		\begin{lemma}\label{lem:wm-admissible-subclass}
			The preference domain \(\mathcal R^{WM}\) is a subclass of the full student-preference domain in the sense of \citet[Theorem 5]{Hatfield2021Stability}.
		\end{lemma}

		\begin{proof}
		For each student \(i\) and school \(s\), let \(\mathcal P_i^s\) be the set of strict preferences over
		\((\mathcal X_i\cap\mathcal X_s)\cup\{\emptyset\}\), and define
		\[
		\mathcal Q_i^s
		:=
    	\left\{
		\succ_i^s\in\mathcal P_i^s:
		(i,s,m')\succ_i^s(i,s,m)
		\text{ whenever }m'>m
    	\right\}.
		\]
		For every nonempty local contract set, \(\mathcal Q_i^s\) contains a preference under which some contract is acceptable. Moreover, truncating a preference in \(\mathcal Q_i^s\) changes only the position of the outside option and preserves the relative ranking of all contracts. Consequently, every such truncation remains in \(\mathcal Q_i^s\). Thus, \(\mathcal Q_i^s\) is a subclass of \(\mathcal P_i^s\) in the sense of \citet{Hatfield2021Stability}.

		Let \(\mathcal Q_i\) be the preference domain generated by \((\mathcal Q_i^s)_{s\in\mathcal S}\). Since within-school monotonicity imposes exactly these school-by-school restrictions and places no restriction on comparisons across different schools,
		\[
		\mathcal Q_i=\mathcal R_i^{WM}.
		\]
		Therefore,
		\[
		\mathcal R^{WM}
		=\prod_{i\in\mathcal I}\mathcal R_i^{WM}
		=\prod_{i\in\mathcal I}\mathcal Q_i
		\]
		is a subclass of the full preference domain.
		\end{proof}

	\end{appendices}

\end{document}